\documentclass{article}
\PassOptionsToPackage{numbers,sort&compress}{natbib[authoryear]}

\usepackage[arxiv]{optional}

\opt{conf}{\usepackage{aistats2019}}
\opt{arxiv}{\usepackage[preprint]{arxiv_2018}}
\usepackage{microtype}
\usepackage{graphicx}
\usepackage{subfigure}
\usepackage{booktabs} 

\usepackage{hyperref}

\usepackage[utf8]{inputenc}
\usepackage{amsmath,amsthm,amssymb}
\usepackage{mathtools}
\usepackage{lipsum,graphicx,multicol}

\makeatletter
\newcommand\footnoteref[1]{\protected@xdef\@thefnmark{\ref{#1}}\@footnotemark}
\makeatother

\usepackage{microtype}
\usepackage{graphicx}
\usepackage{float}
\usepackage{nicefrac,xfrac}
\usepackage{subfigure}
\usepackage{booktabs} 
\usepackage{commath}
\usepackage{bbm}
\usepackage{color}
\usepackage[dvipsnames]{xcolor}

\newcommand{\R}{\mathbb{R}}
\DeclareMathOperator\diag{diag}

\newcommand{\E}{\mathbb{E}}
\newcommand{\like}{\mathcal{L}}
\newcommand{\prior}{\pi_0}
\newcommand\inner[2]{\langle #1, #2 \rangle}

\usepackage{mathtools}
\usepackage[capitalize, sort]{cleveref}
\usepackage{thmtools}

\theoremstyle{plain}
\newtheorem{nthm}{Theorem}[section]
\newtheorem{nprop}[nthm]{Proposition}
\newtheorem{nlem}[nthm]{Lemma}
\newtheorem{ncor}[nthm]{Corollary}

\theoremstyle{definition}
\newtheorem{ndefn}[nthm]{Definition}

\theoremstyle{remark}

\newtheorem*{rmk}{Remark}

\crefname{nlem}{Lemma}{Lemmas}
\crefname{nprop}{Proposition}{Propositions}
\crefname{ncor}{Corollary}{Corollaries}
\crefname{nthm}{Theorem}{Theorems}
\crefname{nexa}{Example}{Examples}
\crefname{ndefn}{Definition}{Definitions}
\crefname{nassum}{Assumption}{Assumptions}
\crefformat{footnote}{#1\footnotemark[#2]#3}

\usepackage{hyperref}

\opt{arxiv}{
\title{The Kernel Interaction Trick: Fast Bayesian Discovery of Pairwise  Interactions in High Dimensions}
\author{
Raj Agrawal \\
CSAIL \\
Massachusetts Institute of Technology \\
\texttt{r.agrawal@csail.mit.edu}
\And
Jonathan H.~Huggins \\
Department of Biostatistics \\
Harvard University \\
\texttt{jhuggins@mit.edu}
\AND
Brian Trippe \\
CSAIL \\
Massachusetts Institute of Technology \\
\texttt{btrippe@mit.edu}
\And
Tamara Broderick \\
CSAIL \\
Massachusetts Institute of Technology \\
\texttt{tbroderick@csail.mit.edu}
}}

\begin{document}

\opt{arxiv}{\maketitle}
\opt{conf}{
\runningtitle{Data-dependent compression of random features \\ for large-scale kernel approximation}%

\twocolumn[

\aistatstitle{Scalable Gaussian Process Inference with \\ Finite-data Mean and Variance Guarantees}

\aistatsauthor{ Author 1 \And Author 2 \And  Author 3 }
\aistatsaddress{ Institution 1 \And  Institution 2 \And Institution 3 }

]
}

%

%






\begin{abstract}
    Discovering interaction effects on a response of interest is a fundamental problem faced in biology, medicine, economics, and many other scientific disciplines. In theory, Bayesian methods for discovering pairwise interactions enjoy many benefits such as coherent uncertainty quantification, the ability to incorporate background knowledge, and desirable shrinkage properties. In practice, however, Bayesian methods are often computationally intractable for even moderate-dimensional problems. Our key insight is that many hierarchical models of practical interest admit a particular Gaussian process (GP) representation; the GP allows us to capture the posterior with a vector of $O(p)$ kernel hyper-parameters rather than $O(p^2)$ interactions and main effects. With the implicit representation, we can run Markov chain Monte Carlo (MCMC) over model hyper-parameters in time and memory \emph{linear} in $p$ per iteration. We focus on sparsity-inducing models and show on datasets with a variety of covariate behaviors that our method: (1) reduces runtime by orders of magnitude over naive applications of MCMC, (2) provides lower Type I and Type II error relative to state-of-the-art LASSO-based approaches, and (3) offers improved computational scaling in high dimensions relative to existing Bayesian and LASSO-based approaches.
             
\end{abstract}

\section{Introduction} \label{sec:intro}

Many decision-making and scientific tasks require understanding how a set of covariates relate to a target response.
For example, in clinical trials and precision medicine, researchers seek to characterize how individual-level traits impact treatment effects, and in modern genomic studies,
researchers seek to identify genetic variants that are risk factors for particular diseases.  While linear regression is a default method for these tasks and many others due to its ease of interpretability, its simplicity often comes at the cost of failing to learn more nuanced information from the data. 
A common way to increase flexibility, while still retaining the interpretability of linear regression, is to augment the covariate space. For instance, two genes together might be highly associated with a disease even though individually they exhibit only moderate association; thus, an analyst might want to consider the multiplicative effect of pairs of covariates co-occurring.  

Unfortunately, augmenting the covariate space by including all possible \emph{pairwise interactions} means the number of parameters to analyze grows quadratically with the number of covariates $p$. This growth leads to many statistical and computational difficulties that are only made worse in  the high-dimensional setting, where $p$ is much larger than the number of observations $N$. And $p \gg N$ is often exactly the case of interest in genomic and medical applications. To address the statistical challenges, practitioners often enforce a sparsity constraint on the model, reflecting an assumption that only a small subset of all covariates affect the response.  The problem of identifying this subset is a central problem in high-dimensional statistics and many different LASSO-based approaches have been proposed to return sparse point estimates. However, these methods do not address how to construct valid confidence intervals or adjust for multiple comparisons\footnote{While the knockoff filter introduced in \citet{barber2015} is a promising way to control the false discovery rate, such a method has not been evaluated theoretically or empirically for interaction models.} \citep{lass_heirch, lim_heirch_lasso, wu_lasso_logistic, pruning_lasso, backtrack_heirch_lasso}.  

Fortunately, hierarchical Bayesian methods have a shrinkage effect, naturally handle multiplicity, can provide better statistical power than multiple comparison corrections \cite{mult_comp}, and can leverage background knowledge. However, naive approaches to Bayesian inference are computationally intractable for even moderate-dimensional problems. This intractability has two sources. The first source can be seen even in the simple case of conjugate linear regression with a multivariate Gaussian prior. Let $\tilde{X}$ denote the augmented data matrix including all pairwise interactions, $\Sigma$ the multivariate Gaussian prior covariance on parameters, and $\sigma^2$ the noise variance. Given $N$ observations, computing the posterior requires inverting $\Sigma^{-1} + \frac{1}{\sigma^2} \tilde{X}^T\tilde{X}$, which takes $O(p^2N^2 + N^3)$ time. The second source is that reporting on $O(p^2)$ parameters simply has $O(p^2)$ cost.

We propose to speed up inference in Bayesian linear regression with pairwise interactions by addressing both problems. In the first case, we show how to represent the original model using a {Gaussian process} (GP). We use the GP kernel in our \emph{kernel interaction sampler} to take advantage of the special structure of interactions and avoid explicitly computing or inverting $\Sigma^{-1} + \frac{1}{\sigma^2} \tilde{X}^T\tilde{X}$. In the second case, we develop a \emph{kernel interaction trick} to compute posterior summaries exactly for main effects and interactions between selected main effects to avoid the full $O(p^2)$ reporting cost. In sum, we show that we can recover posterior means and variances of regression coefficients in $O(pN^2 + N^3)$ time, a $p$-fold speed-up. We demonstrate the utility and efficacy of our general-purpose computational tools for the \emph{sparse kernel interaction model} (SKIM), which we propose in \cref{sec:model} for identifying sparse interactions. In \cref{sec:experiments} we empirically show (1) improved Type I and Type II error relative to state-of-the-art LASSO-based approaches and (2) improved computational scaling in high dimensions relative to existing Bayesian and LASSO-based approaches. Our methods extend naturally beyond pairwise interactions to higher-order \emph{multi-way} interactions, as detailed in \cref{A:higer_degree}.

\section{Preliminaries and Related Work} \label{sec:prelims}

Suppose we observe data $D = \{(x^{(n)}, y^{(n)})\}^N_{n=1}$ with covariates $x^{(n)} \in \R^p$ and responses $y^{(n)} \in \R$. Let $X \in \mathbb{R}^{N \times p}$ denote the design matrix and $Y \in \mathbb{R}^N$ denote the vector of responses. 
Linear models assume that each $y^{(n)}$ is a (noisy) linear function of the covariates $x^{(n)}$.
A common strategy to increase the expressivity of linear models is to augment the original covariates $x^{(n)}$ with their pairwise interactions
$$\Phi_2^T(x) \coloneqq [1, x_1, \cdots, x_p, x_1 x_2, \cdots, x_{p-1}x_p, x_1^2, \cdots, x_p^2].$$
That is, for a parameter $\theta \in \R^{p(p+1)/2}$ and zero-mean i.i.d.\ errors $\epsilon^{(n)}$, we assume the data are generated according to
\begin{equation} \label{eq:interaction_likelihood}
y^{(n)} = \theta^T \Phi_2(x^{(n)}) + \epsilon^{(n)}.
\end{equation}
%
%
%

Our goal is to identify which interaction terms have a significant effect on the response. Detecting such interactions is important for many applications. For example, in genomics, two-way interaction terms are needed to detect possible epistasis between genes \citep{gwas_interaction, epistatis_gene2} and to appropriately account for the site- and sample-specific effects of GC content on genomic and other types of sequencing data \citep{gc_content,Risso2011}. 
In economics and clinical trials, pairwise interactions between covariates and treatment 
are used to estimate the heterogeneous effect a treatment has across different subgroups \citep[Section 6]{clinical_trials}.
 Unfortunately, having $O(p^2)$ parameters creates statistical and computational challenges when $p$ is large.
 
To address the statistical issues, practitioners often assume that $\theta$ is \emph{sparse} (i.e., contains only a few non-zero values), and that $\theta$ satisfies \emph{strong hierarchy}. That is, an interaction effect $\theta_{x_i x_j}$ is present only if both of the main effects $\theta_{x_i}$ and $\theta_{x_j}$ are present, where $\theta_{x_i x_j}$ and $\theta_{x_i}$ are the regression coefficients of the variables $x_i x_j$ and $x_i$ respectively \citep{lass_heirch, lim_heirch_lasso,wu_lasso_logistic,pruning_lasso, chipman_bayes_glm}. By assuming such low-dimensional structure, inference tasks such as parameter estimation and {variable selection} become more tractable statistically. However, sparsity constraints create computational difficulties. For example, finding the {maximum-likelihood estimator} (MLE) subject to $\| \theta \|_0 \leq s$ requires searching over $\Theta(p^{2s})$ active parameter subsets. To avoid the combinatorial issues resulting from an $L_0$ penalty, recent works \citep{lass_heirch, lim_heirch_lasso} have instead used $L_1$ penalties to encourage parameter sparsity for interaction models; $L_1$ penalties have a long history in high-dimensional linear regression \citep{atom_pursuit,dantzig,lasso},

Maximizing the likelihood with an added $L_1$ penalty is a convex problem. But each iteration of a state-of-the-art solver for methods given by \citet{lass_heirch} and \citet{lim_heirch_lasso} still takes $O(Np^2)$ time. To handle larger $p$, \citet{wu_lasso_logistic,pruning_lasso, backtrack_heirch_lasso} have proposed various pruning heuristics for finding locally optimal solutions. However, since these methods do not provide an exact solution to the optimization problem, any statistical guarantees (such as the statistical rate at which these estimators converge to the true parameter as a function of $N$ and $p$) are weaker than those for exact methods.  

$L_1$-based methods face a number of additional challenges: 
constructing valid confidence intervals, incorporating background knowledge, and controlling for the issue of multiple comparisons when testing many parameters for statistical significance. 
In many applications such as genome-wide association studies, controlling for multiplicity is critical to prevent wasting resources on false discoveries. Moreover, since $\dim(\Phi_2) = p(p+1)/2$, $\theta$ can be very high dimensional even when $p$ is moderately large. Hence, there will typically be nontrivial uncertainty when attempting to estimate $\theta$. Fortunately, hierarchical Bayesian methods have (1) a natural shrinkage or regularization effect such that multiple testing corrections are no longer necessary, (2) better statistical power than using multiple comparison correction terms such as Bonferroni \citep{mult_comp}, and (3) naturally provide calibrated uncertainties.  
Bayesian methods can also incorporate expert information. 

Though they offer desirable statistical properties, Bayesian approaches are computationally expensive. Previous efforts \citep{heirc_sparisty,chipman_bayes_glm} have focused on developing \emph{hierarchical sparsity priors} that  promote strong hierarchy, analogous to the LASSO-based approaches \citep{lass_heirch, lim_heirch_lasso,wu_lasso_logistic,pruning_lasso}. But these methods do not address the computational intractability of inference for even moderate-dimensional problems. 

We address the computational challenges of inference by developing the \emph{kernel interaction trick} (\cref{sec:compute_expect}), which allows us to access posterior marginals of $\theta$ without ever representing $\theta$ explicitly. Note that while some previous works have used a degree-two polynomial kernel to implicitly generate all pairwise interactions \citep{kernel_genome, weissbrod_kernel, su_kernel}, those works have focused on prediction or estimating the cumulative proportion of variance explained by interactions rather than our present focus on posterior inference.



\section{Bayesian Models with Interactions} \label{sec:problem_setup}
Our goal is to estimate and provide uncertainties for the parameter $\theta \in \R^{\dim(\Phi_2)}$. To take a Bayesian approach, we encode the state of knowledge before observing the data $D$ in a \emph{prior} $\prior(\theta)$.  We express the \emph{likelihood} as $\like(Y \mid \theta, X) = \prod_{n=1}^{N} \like(y^{(n)} \mid \theta, x^{(n)})$. Applying Bayes' theorem yields the \emph{posterior} distribution $\pi(\theta \mid D) \propto \like(Y \mid \theta, X) \prior(\theta)$, which describes the state of knowledge about $\theta$ after observing the data $D$. For a function $f$ of interest, we wish to compute the posterior expectation
\begin{equation} \label{eq:expectation_full}
\E_{\pi(\theta \mid D)}[f(\theta)] = \int f(\theta) \pi(\theta \mid D)d\theta.
\end{equation}  
Typically, $f(\theta) = \theta_j$ or $f(\theta) = \theta_j^2$, which together allow us to compute the {posterior mean} and {variance} of each $\theta_j$. 

\noindent \textbf{Generative model.} Going forward, we model $\theta$ as being drawn from a \emph{Gaussian scale mixture} prior to encode desirable properties such as sparsity and strong hierarchy \citep[{cf.}][]{heirc_sparisty, chipman_bayes_glm, george1993variable, horse_prior, finnish_prior}. These priors have also been used beyond sparse Bayesian regression \citep[{cf.}][]{scale_mix_interest,scale_mix_images,scale_mix_insurance}.
A Gaussian scale mixture is equivalent to assuming that there exists an auxiliary random variable $\tau \sim p(\tau)$ such that $\theta$ is conditionally Gaussian given $\tau$. Let $\Sigma_{\tau}$ denote the covariance matrix for $p(\theta \mid \tau)$. Also, let $\sigma^2$ be the latent noise variance in the likelihood; since it is typically unknown, we treat it as random and put a prior on it as well. Hence, the full generative model can be written
\begin{equation} \label{eq:gen_model}
\begin{split}
\tau &\sim p(\tau) \\
\sigma^2 &\sim p(\sigma^2 ) \\
\theta \mid \tau & \sim \mathcal{N}(0, \Sigma_{\tau}) \\
y^{(n)} \mid x^{(n)}, \theta, \sigma^2 & \sim \mathcal{N}(\theta^T\Phi_2(x^{(n)}), \sigma^2).
\end{split}
\end{equation}
\begin{table}[!t]
    \caption{Per-iteration MCMC runtime and memory scaling of methods for sampling two-way interactions. {NAIVE} refers to explicitly factorizing $\Sigma_{N, \tau}$ to compute $p(D \mid \tau, \sigma^2)$, {WOODBURY} refers to using the Woodbury identity and matrix determinant lemma to compute $p(D \mid \tau, \sigma^2)$, and {FULL} refers to jointly sampling $\theta$ and $\tau$. The third column provides the number of parameters sampled.}
\label{table:bayes_runtime_comp}
\vskip 0.15in
\begin{center}
\begin{small}
\begin{sc}
\begin{tabular}{lcccr}
\toprule
Method & Time & Memory &  \# \\ 
\midrule
Our Method    & $O(pN^2 + N^3)$ & $O(pN + N^2)$ & $O(p)$ \\
Naive & $O(p^6 + p^4N)$ & $O(p^4 + p^2N)$ & $O(p)$\\
Woodbury     & $O(p^2N^2 + N^3)$ & $O(p^2 N  + N^2)$ & $O(p)$\\
Full                  & $O(p^2N)$ & $O(p^2N)$ & $\Theta(p^2)$ \\
\bottomrule
\end{tabular}
\end{sc}
\end{small}
\end{center}
\vskip -0.1in
\end{table}

\noindent \textbf{Computational challenges of inference.}
Again, our main goal is to tractably compute expectations of functions under the {posterior} $\pi(\theta \mid D) \propto \like(Y \mid \theta, X)\prior(\theta)$. Since there are $\Theta(p^2)$ parameter components, direct numerical integration over each of these components is feasible only when $p$ is at most 3 or 4.  As a result we turn to Monte Carlo integration.  Two natural Monte Carlo estimators one might use to approximate $\E_{\pi(\theta \mid D)}[f(\theta)]$ are
\begin{enumerate} \label{eq:possible_ways}
\item $\frac{1}{T} \sum_{t=1}^T f(\theta^{(t)})$ with $\theta^{(t)} \overset{\text{iid}}{\sim} \pi(\theta \mid D)$ or
\item $\frac{1}{T} \sum_{t=1}^T \E_{\pi(\theta \mid D, \tau^{(t)})}[f(\theta)]$ with $\tau^{(t)} \overset{\text{iid}}{\sim} p(\tau \mid D)$.
\end{enumerate}

For the first estimator, we can use {Markov chain Monte Carlo} (MCMC) techniques to sample each $\theta^{(t)}$ approximately independently from $\pi(\theta \mid D)$ since the posterior is available up to a multiplicative normalizing constant. Computing the prior $p(\theta)$, however, may be analytically intractable because it requires marginalizing out $\tau$.  We could instead additionally sample $\tau$.
To use gradient-based MCMC samplers, sampling $\tau$ would require computing the pdfs (and gradients) of the likelihood terms $\like(y^{(n)} \mid x^{(n)}, \theta, \sigma^2)$ and the prior terms $p(\theta \mid \tau)$ and $p(\tau)$. So the cost would be $O(p^2N+ \dim(\tau))$ time per iteration.  Even for $p$ moderately large, the $\Theta(p^2 + \dim(\tau))$ number of parameters might require many MCMC iterations to properly explore such a large space \citep{Mackay1998,mala_dim, hmc_dim}; see also \cref{fig:runtime_full} for an empirical demonstration.  

To explore a smaller space, and hence potentially reduce the number of MCMC iterations required for the chains to {mix}, we might take the second approach: sampling from $p(\tau \mid D)$ by marginalizing out the high-dimensional parameter $\theta$. Sampling each $\tau$ requires computing
\begin{equation} \label{eq:issue_likelihood}
p(D \mid \tau, \sigma^2) = \int p(D \mid \theta, \sigma^2) dp(\theta \mid \tau).
\end{equation}
Since $p(\theta \mid \tau)$ is a multivariate Gaussian density function, $p(D \mid \tau, \sigma^2)$ equals
\begin{equation} \label{eq:norm_constant}
 \frac{(1/\sqrt{2\pi\sigma^2})^N \det(2\pi \Sigma_{N, \tau})^{\frac{1}{2}}  \exp \left(-\frac{1}{2\sigma^2} Y^TY \right)}{\det(2\pi \Sigma_{\tau})^{\frac{1}{2}} \exp \left(-\frac{1}{2\sigma^4} Y^T \Phi_2(X) \Sigma_{N, \tau} \Phi_2(X)^T Y  \right)}
\end{equation}
where $\Sigma_{N, \tau}^{-1} \coloneqq \Sigma^{-1}_{\tau} + \frac{1}{\sigma^2} \Phi_2(X)^T \Phi_2(X)$. 
Unfortunately, computing \cref{eq:norm_constant} 
naively takes prohibitive $O(p^{6} + p^4 N)$ time -- or $O(p^2 N^2 + N^3)$ time when using linear algebra identities; see \cref{table:bayes_runtime_comp} and \cref{sec:woodbury} for details.

\section{The Kernel Interaction Sampler} \label{sec:method}



Our \emph{kernel interaction sampler} (KIS) provides a recipe for efficiently sampling from $p(\tau , \sigma^2 \mid D)$ using MCMC.
Recall from the last section that the computational bottleneck for sampling $\tau$ was computing $p(D \mid \tau, \sigma^2)$, so we focus on that problem here.
We achieve large computational gains (\cref{table:bayes_runtime_comp}) by using the special model structure and a kernel trick to avoid the factorization of $\Sigma_{N, \tau}$ in \cref{eq:norm_constant}. To that end, KIS has three main parts: (1) we re-parameterize the generative model given in \cref{eq:gen_model} using a Gaussian process (GP); (2) we show how to cheaply compute the GP kernel; and (3) we show how these steps translate into computation of $p(D \mid \tau, \sigma^2)$ in time linear in $p$. In \cref{A:higer_degree} we extend to the case of higher-order interactions.


For the moment, suppose that we could construct a covariance function $k_{\tau}$ such that the generative model in \cref{eq:gen_model} could be re-parameterized as:
\begin{equation} \label{eq:gen_gp_model}
\begin{split}
\tau &\sim p(\tau) \\
g \mid \tau & \sim GP(0, k_{\tau}) \\
\sigma^2 &\sim p(\sigma^2) \\
y^{(n)} \mid g, x^{(n)}, \sigma^2 & \sim \mathcal{N}(g(x^{(n)}), \sigma^2),
\end{split}
\end{equation}
where $GP(0, k_{\tau})$ denotes a Gaussian process (GP) with mean function zero. Defining the kernel matrix $(K_{\tau})_{ij} \coloneqq k_{\tau}(x^{(i)}, x^{(j)})$, 
we can conclude that~\citep[see][Eq.~2.30]{gp_book}
\begin{equation} \label{eq:marg_like}
\log p(D \mid \tau, \sigma^2) = -\frac{1}{2}Y^T L^{-1} Y - \frac{1}{2} \log|L| - \frac{N}{2} \log 2\pi, 
\end{equation}
where $L$ equals $K_{\tau} + \sigma^2 I_{N}$. Let $T_k$ denote the time it takes to evaluate $k_{\tau}$ on a pair of points. The computational bottleneck of \cref{eq:marg_like} is computing and factorizing $K_{\tau}$, which take $O(N^2 T_k)$ and $O(N^3)$ time, respectively. Hence, as long as $T_k$ is $O(p)$, we can compute $p(D \mid \tau, \sigma^2)$ in time \emph{linear} in $p$. To achieve this scaling, we first show (in the next result) that \emph{any} generative model in the form of \cref{eq:gen_model} can be rewritten in the form of \cref{eq:gen_gp_model}. We then show how $k_\tau$ can be evaluated in $O(p)$ time for the models of interest.
\begin{nprop} \label{prop:explicit_kernel} (Gaussian process representation)
Let $Y$ and $\tilde{Y}$ be response vectors generated according to the models in \cref{eq:gen_model} and \cref{eq:gen_gp_model} respectively for design matrix $X \in \R^{N\times p}$. Let $k_{\tau}(x^{(i)}, x^{(j)}) = \Phi_2(x^{(i)})^\top \Sigma_{\tau} \Phi_2(x^{(j)})$. Then, $Y \mid X  \overset{d}{=} \tilde{Y} \mid X$, where $ \overset{d}{=}$ denotes equality in distribution. Moreover, for every draw $g \mid \tau \sim \mathcal{N}(0, k_{\tau})$, there exists some $\theta \in \R^{\dim(\Phi_2)}$ such that $g(\cdot) = \theta^T \Phi_2(\cdot)$. 
\end{nprop}
The proof follows directly by considering the {weight-space view} of a GP \citep[Chapter 2]{gp_book}; see \cref{A:proofs} for details.

Next, we need to show that $k_{\tau}$ can be evaluated in $O(p)$ time for models of interest. This fact is not obvious; computing $k_{\tau}$ on a pair of points naively still requires explicitly computing the high-dimensional feature maps $\Phi_2$ and prior covariance matrix $\Sigma_{\tau}$.
To compute $k_{\tau}$ efficiently, we rewrite it as a weighted sum of {polynomial kernels} of the form
\begin{equation*}
k_{\text{poly}, d}^c(x, \tilde{x}) \coloneqq \left( x^T \tilde{x} + c \right)^d,
\end{equation*}
which each take $O(p)$ time to compute.
Below we define \emph{two-way interaction kernels} as particular linear combinations of these polynomial kernels.
Then we provide a result motivating this class; namely, we show that any diagonal $\Sigma_{\tau}$ prior can be written as a two-way interaction kernel. Fortunately, to the best of our knowledge, all previous high-dimensional Bayesian regression models assume $\Sigma_{\tau}$ is diagonal \citep[{cf.}][]{heirc_sparisty, chipman_bayes_glm, george1993variable, horse_prior, finnish_prior}. Hence, this restriction on $\Sigma_{\tau}$ is mild.
%
\begin{ndefn} \label{def:two_way} (Two-way interaction kernel) We call the kernel $k$ a \emph{two-way interaction kernel} if for some choice of $M_1, M_2 \in \mathbb{N}$, $\alpha, \psi, \lambda^{(m)} \in \R^p_+~(m=1,\dots,M_{1})$, $\nu^{(m)} \in \R_+~(m=1,\dots,M_{2})$, $1 \leq i_m < j_m \leq p~(m=1,\dots,M_{2})$, and $A \in \R$, the kernel $k(x,\tilde{x})$ is equal to
\begin{equation*}
\begin{split}
& \sum_{m=1}^{M_1} k^1_{poly, 2}(\lambda^{(m)} \odot x, \lambda^{(m)} \odot \tilde{x}) + \sum_{m=1}^{M_2} \nu^{(m)} x_{i_m}x_{j_m} \tilde{x}_{i_m} \tilde{x}_{j_m} \\ & +  k^A_{poly, 1}(\alpha \odot x, \alpha \odot \tilde{x})  +  k^0_{poly, 1}(\psi \odot x \odot x, \psi \odot \tilde{x} \odot \tilde{x}),
\end{split}
\end{equation*}
where $\odot$ is the entrywise product.
\end{ndefn}
 
\begin{nthm}[1-to-1 correspondence with diagonal $\Sigma_{\tau}$] \label{thm:induced_prior} 
Suppose $k$ is a two-way interaction kernel. Then
\begin{equation}
k(x, \tilde{x}) = \Phi_2(x)^\top S \Phi_2(\tilde{x}),
\end{equation}
where the induced prior covariance matrix $S$ is diagonal. The entries of $S$ are given by
\begin{equation*}
\begin{split}
 \diag(S)_{(i)} &= \alpha_i^2 + 2 \sum_{m=1}^{M_1} \left[ \lambda^{(m)}_i \right]^2 \\
  \diag(S)_{(ij)} &= 2 \sum_{m=1}^{M_1} \left[ \lambda^{(m)}_i  \lambda^{(m)}_j \right]^2 + \sum_{k: i_k = i, j_k=j}^{M_2} \nu^{(m)}  \\
 \diag(S)_{(ii)} &= \psi_i^2 + \sum_{m=1}^{M_1} \left[ \lambda^{(m)}_i  \right]^4 \\
 \diag(S)_{(0)} &= M_1 + A,
\end{split}
\end{equation*}
where $\diag(S)_{(i)} $, $\diag(S)_{(ij)}$, $\diag(S)_{(ii)}$, and $\diag(S)_{(0)}$ denote the prior variances of the main effect $\theta_{x_i}$, interaction effect $\theta_{x_ix_j}$, quadratic effect $\theta_{x_i^2}$, and intercept $\theta_{0}$, respectively. 
Furthermore, for any diagonal covariance matrix $S \in \R^{\dim(\Phi_2) \times \dim(\Phi_2)} $, there exists a two-way interaction kernel that induces $S$ as a prior covariance matrix.
\end{nthm}
\cref{thm:induced_prior} (proof in \cref{A:mult_kern_trick}) and \cref{prop:explicit_kernel}  imply that two-way interaction kernels induce a space of models in 1-to-1 correspondence with models in the form of \cref{eq:gen_model} when $\Sigma_{\tau}$ is constrained to be diagonal. 
Since most models of practical interest have $\Sigma_{\tau}$ diagonal, we can readily construct the two-way interaction kernel corresponding to $\Sigma_{\tau}$
by solving the system of equations
\begin{equation} \label{eq:system_eq}
\begin{split}
\diag(S)_{(i)} = \diag(\Sigma_{\tau})_{(i)} \ &\ \ \diag(S)_{(ij)} =  \diag(\Sigma_{\tau})_{(ij)} \\
\diag(S)_{(ii)} = \diag(\Sigma_{\tau})_{(ii)} &\ \ \diag(S)_{(0)} \ =  \diag(\Sigma_{\tau})_{(0)}
\end{split}
\end{equation}
Each of the $M_1 + 2$ polynomial kernels takes $O(p)$ time to compute, and each of the $M_2$ product terms takes $O(1)$ time. Therefore, we want to select $M_1$ and $M_2$ small so that $k_\tau$ can be computed quickly. Since there are more degrees of freedom (i.e., free variables) available to solve \cref{eq:system_eq} as $M_1$ and $M_2$ increase, eventually a solution will exist as we show in \cref{A:mult_kern_trick}. But \cref{thm:induced_prior} does not tell us how large $M_1$ and $M_2$ have to be for an arbitrary model. In \cref{A:example_models}, we solve \cref{eq:system_eq} for a variety of models of practical interest and show that in these cases, $M_1$ and $M_2$ can be set very small (between one and three). Thus $k_\tau$ can be computed in $O(p)$ time, and so the kernel matrix $K_{\tau}$ can be computed in $O(N^2p)$ time. Finally, then, we may compute the likelihood $p(D \mid \tau, \sigma^2)$ in $O(N^2p + N^3)$ time.


\section{The Kernel Interaction Trick: Recovering Posterior Marginals} \label{sec:compute_expect}

Even if we are able to sample $\tau$ much faster using KIS, the problem of computing $\E_{p(\theta \mid D, \tau)}[f(\theta)]$ remains unresolved. In this section, we show that, given $K_\tau$, any such expectation can be recovered in $O(1)$ time by evaluating the GP posterior at certain test points.

To provide the main intuition for our solution, suppose we would like to compute the posterior mean of the main effect $\theta_{x_i}$. Let $e_i \in \R^p$ denote the $i$th unit vector. Since $g = \theta^T\Phi_2$ by \cref{prop:explicit_kernel}, we have
\begin{equation} \label{eq:intuition_test_pt}
\begin{split}
g(e_i) &= \theta_{x_i} + \theta_{x_i^2} \\
g(-e_i) &= -\theta_{x_i} + \theta_{x_i^2} \\
\frac{g(e_i) - g(-e_i)}{2} &= \theta_{x_i}.
\end{split}
\end{equation}
Since $g$ is a Gaussian process, the distribution of
%
$
Z_g \coloneqq (g(e_i), g(-e_i)) \mid D,  \tau 
$
%
is multivariate Gaussian and can be computed in closed form by appropriate matrix multiplications of the kernel matrix $K_{\tau}$; see \cref{prop:formula_marginal} below for details. Then, by consulting \cref{eq:intuition_test_pt}, one can recover $\theta_{x_i} \mid D, \tau$ as the linear combination $[\nicefrac{1}{2}, -\nicefrac{1}{2}]^T Z_g \mid D, \tau$, which is univariate Gaussian.
While we have focused on a particular instance here, this example provides the main insight for the general formula to compute $\E_{p(\theta \mid D, \tau)}[f(\theta)]$  from $K_{\tau}$.
\begin{nthm} (The kernel interaction trick) \label{prop:formula_marginal}
Let $H_{\tau} \coloneqq (K_{\tau} + \sigma^2 I_N)^{-1}$ and
\begin{equation*}
    A_{ij} \coloneqq [e_i, -e_i, e_j, e_{i}+e_{j}]^T \in \R^{4 \times p}.
\end{equation*}
Let $K_{\tau}(A_{ij}, X) = K_{\tau}(X, A_{ij})^T$ be the $4 \times N$ matrix formed by taking the kernel between each row of $A_{ij}$ with each row of $X$.
For a row-vector $a \in \R^4$, define the scalars $\mu_a \coloneqq a K_{\tau}(A_{ij}, X) H_{\tau} Y$ 
and
\begin{equation*}
    \sigma_a^2 \coloneqq a \left[K_{\tau}(A_{ij}, A_{ij}) - K_{\tau}(A_{ij}, X) H_{\tau} K_{\tau}(X, A_{ij})  \right]a^T.
\end{equation*}
Then the distributions of $\theta_{x_i} \mid \tau, D$, $\theta_{x_ix_j} \mid \tau, D$, and $\theta_{x_i^2} \mid \tau, D$ are given by $\mathcal{N}(\mu_a, \sigma_a^2)$ with, respectively, $a = (\nicefrac{1}{2}, -\nicefrac{1}{2}, 0, 0)$, $a = (-\nicefrac{1}{2}, \nicefrac{1}{2}, -1, 1)$, and $a = (\nicefrac{1}{2}, \nicefrac{1}{2}, 0, 0)$.
%
%
%
%
\end{nthm}
\begin{ncor} \label{cor:get_effects_time}
Given $K_{\tau}$, the distributions of $\theta_{x_{i}}$, $\theta_{x_{i}x_{j}}$, and $\theta_{x_{i}^2}$  take $O(1)$ additional time and memory to compute.
\end{ncor}

We prove \cref{prop:formula_marginal} and \cref{cor:get_effects_time} in \cref{A:proofs}. In \cref{A:the_mult_kern_trick}, we generalize \cref{prop:formula_marginal} by showing how to obtain the joint posterior distribution of any subset of parameters contained in $\theta$. Hence, we can compute $\E_{p(\theta \mid D, \tau)}[f(\theta)]$ for an arbitrary $f$ using the kernel interaction trick.

Note that if we would like to obtain the posterior mean of all $\Theta(p^2)$ parameters, then clearly a linear time algorithm in $p$ is impossible.
Instead, we can adopt a lazy evaluation strategy where we compute the posterior of one of the $\Theta(p^2)$ parameters only when it is needed. This approach is effective in the many applications where we do not need to look at all the interactions. In particular, we might first find the top $k$ main effects. After selecting these variables, we could examine their interactions. The number of interactions among the main effects (which is $\Theta(k^2)$) is much smaller than the total number of possible interactions (which is $\Theta(p^2)$) if $k \ll p$. Such a strategy is natural if we believe that $\theta$ satisfies the (commonly assumed) strong hierarchy restriction.

\section{SKIM: \underline{S}parse \underline{K}ernel \underline{I}nteraction \underline{M}odel} \label{sec:model}


To demonstrate the utility and efficacy of the kernel interaction sampler and kernel interaction trick, we choose a particular model that we call the \emph{sparse kernel interaction model (SKIM)}. In what follows, we first detail SKIM, which we will see promotes sparsity and strong hierarchy. Then, by observing that SKIM is a special case of the general model in \cref{eq:gen_model}, we can show that SKIM induces a two-way interaction kernel via \cref{thm:induced_prior} and \cref{eq:system_eq}. We will see that this kernel has only 3 components and thus takes only $O(p)$ time to evaluate. By \cref{cor:general_kern}, we can compute the distribution of interaction terms from SKIM in $O(1)$ time once we have computed the kernel matrix. Hence, the final computation time for discovering main effects and interaction effects with SKIM will be $O(N^2 p + N^3)$ by the discussion at the end of \cref{sec:compute_expect}.


SKIM\footnote{See \url{https://github.com/agrawalraj/skim} for the code.} is given in full detail, together with discussion of hyperparameter selection and intepretation, in \cref{A:skim_details}. It is a particular instance of a general class of hierarchical sparsity priors~\citep[{cf.}][]{heirc_sparisty, chipman_bayes_glm, george1993variable} that have the following form:
\begin{equation} \label{eq:hierch_prior}
\begin{split}
\kappa \sim p(\kappa) \quad \eta &\sim p(\eta) \quad c^2 \sim p(c^2) \\
\theta_{x_i} \mid \kappa, \eta &\sim \mathcal{N}(0, \eta_1^2 \kappa_i^2)  \\
\theta_{x_i x_j} \mid \kappa, \eta &\sim \mathcal{N}(0, \eta_2^2 \kappa_i^2 \kappa_j^2) \\
\theta_{x_i^2} \mid \kappa, \eta &\sim \mathcal{N}(0, \eta_3^2 \kappa_i^4) \\
\theta_0 \mid c^2 &\sim \mathcal{N}(0, c^2),
\end{split}
\end{equation}
where $\theta_0$ is the intercept term and every $\eta_i$ or $\kappa_j$ is a scalar. 

We next show that any prior in the form of \cref{eq:hierch_prior} induces a $O(p)$ two-way interaction kernel. The proof is in \cref{A:proof_of_skim}.
\begin{nprop} \label{prop:pairwise_sparse_kern} 
Taking $\tau \coloneqq (\eta, \kappa, c^2)$, the generative model in \cref{eq:hierch_prior} is equivalent to using the following kernel in \cref{eq:gen_gp_model}:
\begin{equation*} 
\begin{split}
   k_{\tau}(x, \tilde{x}) &= \frac{\eta_2^2}{2} k_{\text{poly}, 2}^{1}(\kappa \odot x, \kappa \odot \tilde{x}) \\
    &\phantom{=~} (\eta_3^2 - \frac{\eta_2^2}{2})k_{\text{poly}, 1}^{0}(\kappa \odot x \odot x, \kappa \odot \tilde{x} \odot \tilde{x}) \\ 
    &\phantom{=~} + \left (\eta_1^2 - \eta_2^2 \right) k_{\text{poly}, 1}^{0}(\kappa \odot x , \kappa \odot \tilde{x}) + c^2 - \frac{\eta_2^2}{2}.
\end{split}
\end{equation*}
\end{nprop}

\section{Experiments} \label{sec:experiments}
\noindent \textbf{Time and memory cost versus Bayesian baselines.} 
We first assess the computational advantages of our kernel interaction sampler (KIS) by comparing it with each baseline Bayesian method in \cref{table:bayes_runtime_comp}. We start by profiling the time and memory cost of computing $p(D \mid \tau, \sigma^2)$, which we have seen is a computational bottleneck for sampler option 2 in \cref{sec:problem_setup}. 
In \cref{fig:runtime}, we depict the time and memory cost of $p(D \mid \tau, \sigma^2)$ computation for conjugate linear regression with an isotropic Gaussian prior on synthetic datasets with $N=50$. 
We vary $p$ but not $N$ because we are interested primarily in the high-dimensional case when $p$ is large relative to $N$. 
\cref{fig:runtime} shows that KIS yields orders-of-magnitude speed and memory improvements over the baseline methods for computing $p(D \mid \tau, \sigma^2)$. 
\begin{figure}[t]
\vspace*{-.1in}
\centering
\begin{multicols}{2}
	\hspace*{-.15in} 
    \subfigure[Runtime complexity]{\includegraphics[width=\linewidth]{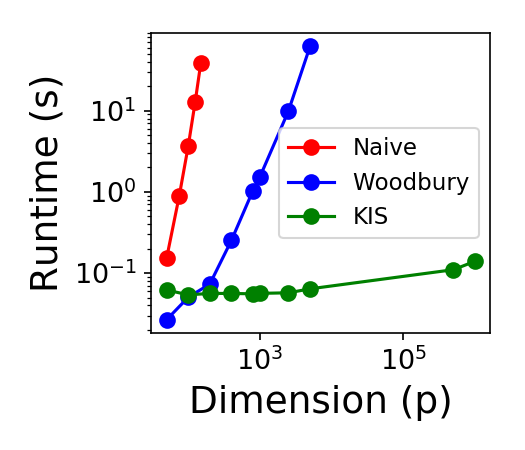}} \par 
    \hspace*{-.15in}  
    \subfigure[Memory complexity]{\includegraphics[width=\linewidth]{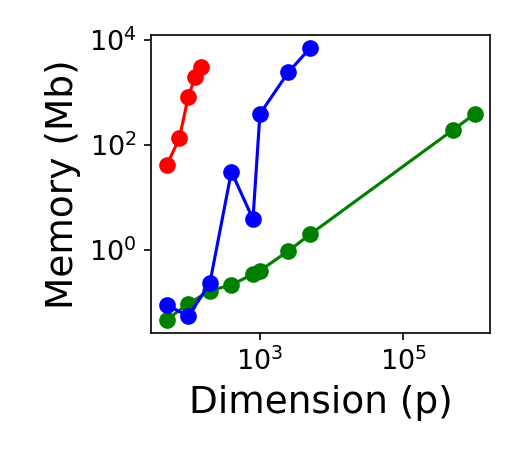}} \par
\end{multicols} 
\vspace*{-.2in}
\caption{
    Empirical evaluation of (a) time and (b) memory scaling with dimension of marginal likelihood computation.
    Woodbury and Naive refer to the baselines in \cref{sec:problem_setup}.} \label{fig:runtime}
\vspace*{-.2in}
\end{figure}

We next compare inference for SKIM using KIS, which marginalizes out $\theta$ and samples $\tau$, to jointly sampling $(\theta, \tau)$ (denoted FULL).\footnote{See the discussion of sampler option 1 in \cref{sec:problem_setup}.} We implemented KIS and FULL in \texttt{Stan} \cite{stan} and used the NUTS algorithm \cite{no_u_turn} for sampling (4 chains with 1,000 iterations per chain). As shown in \cref{fig:runtime_full}(a), KIS is orders of magnitude faster even for lower values of $p$. In \cref{sec:problem_setup} we remarked that since FULL explores a much higher-dimensional space, there might be issues with mixing. To explore this possibility empirically, we check the Gelman--Rubin statistic ($\hat{R}$) values of the output from both KIS and FULL. We found that, for FULL, the $\hat{R}$ values were greater than 1.05, with some reaching as high as 1.5 (indicating poor mixing), while for KIS all $\hat{R}$ values were less than 1.05 (suggesting good mixing). 
\begin{figure}[t]
\centering
\vspace*{-.1in}
\begin{multicols}{2}
	\hspace*{-.15in}
    \subfigure[NUTS runtimes]{\includegraphics[width=1.1\linewidth]{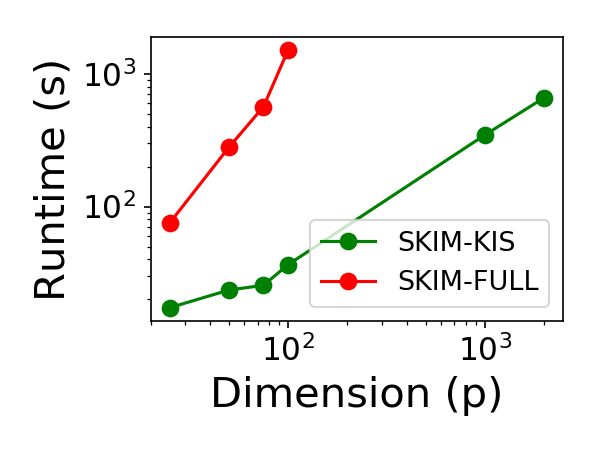}} 
     \par 
     \hspace*{-.15in}  
    \subfigure[LASSO runtime comparisons]{\includegraphics[width=1.1\linewidth]{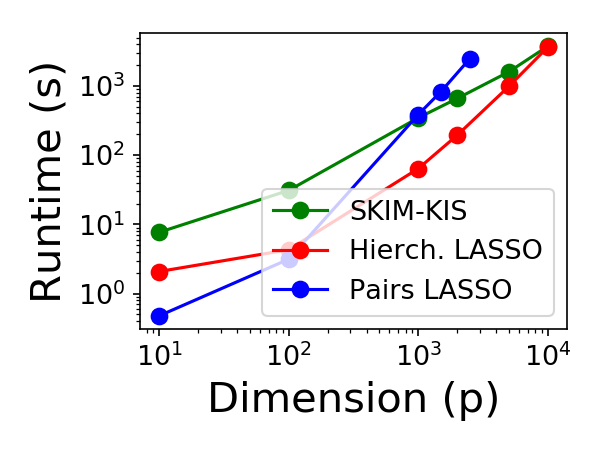}} \par 
\end{multicols}
\vspace*{-.2in}
\caption{The left-hand figure indicates the time to complete four parallel chains of 1000 iterations of NUTS for the SKIM model proposed in \cref{sec:model} using KIS (denoted as SKIM-KIS) and FULL. For each point, KIS had $\hat{R} < 1.05$ while FULL always had $\hat{R} > 1.05$. The right-hand figure compares the runtime of inference for SKIM-KIS versus fitting LASSO-based methods.}
\vspace*{-.2in}
\label{fig:runtime_full}
\end{figure}
\begin{figure}[t]
\vspace*{-.1in}
\centering
\begin{multicols}{2}
 	\hspace*{-.15in}
    \subfigure[Main effects]{\includegraphics[width=1.05\linewidth]{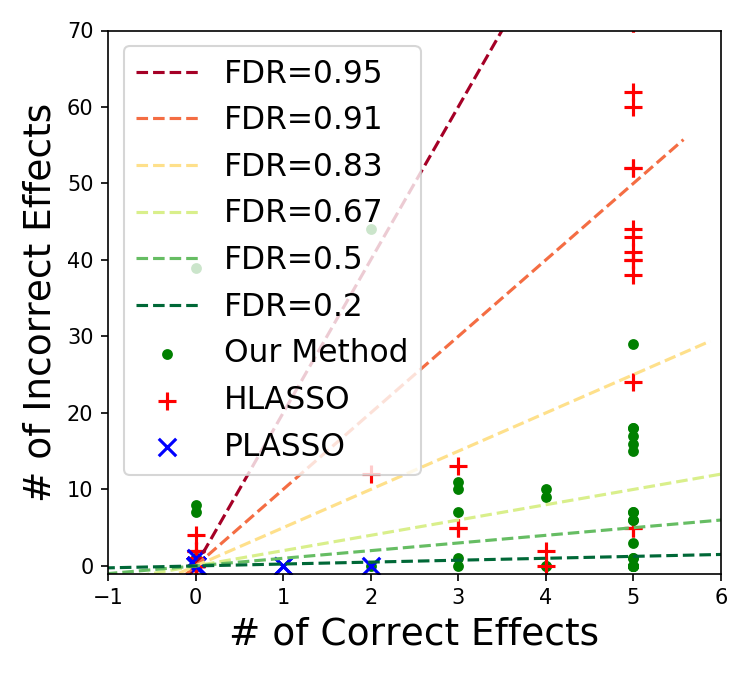}}\par 
    \hspace*{-.3in} 
    \subfigure[Main effects differences]{\includegraphics[width=1.05\linewidth]{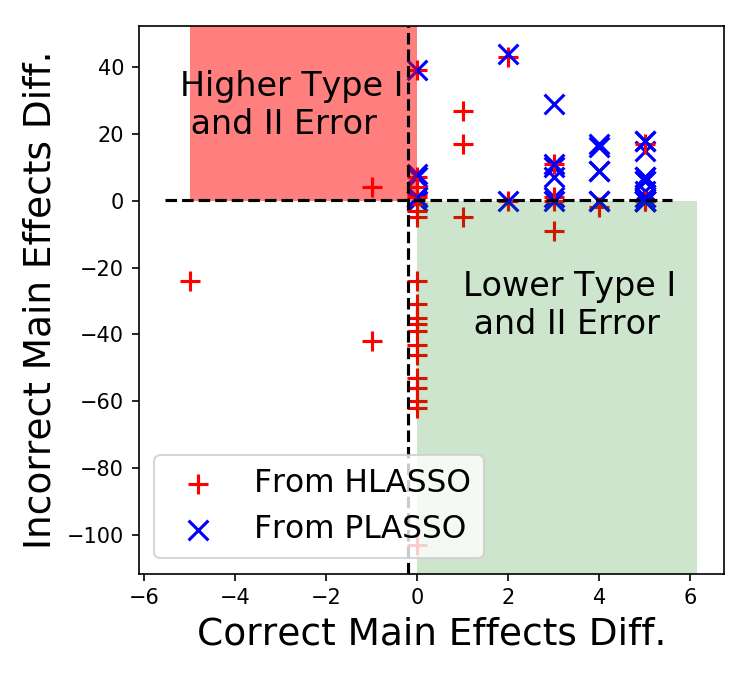}}\par 
    \hspace*{-.2in} 
\end{multicols}
	\vspace*{-.4in}
\begin{multicols}{2}
 	\hspace*{-.15in}
    \subfigure[Pairwise effects]{\includegraphics[width=1.05\linewidth]{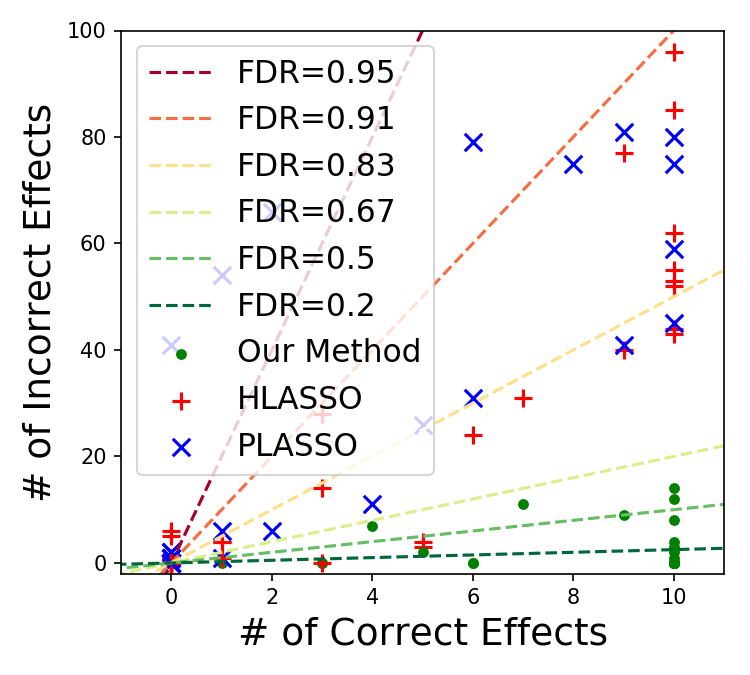}}\par
    \hspace*{-.3in} 
    \subfigure[Pairwise effects differences]{\includegraphics[width=1.05\linewidth]{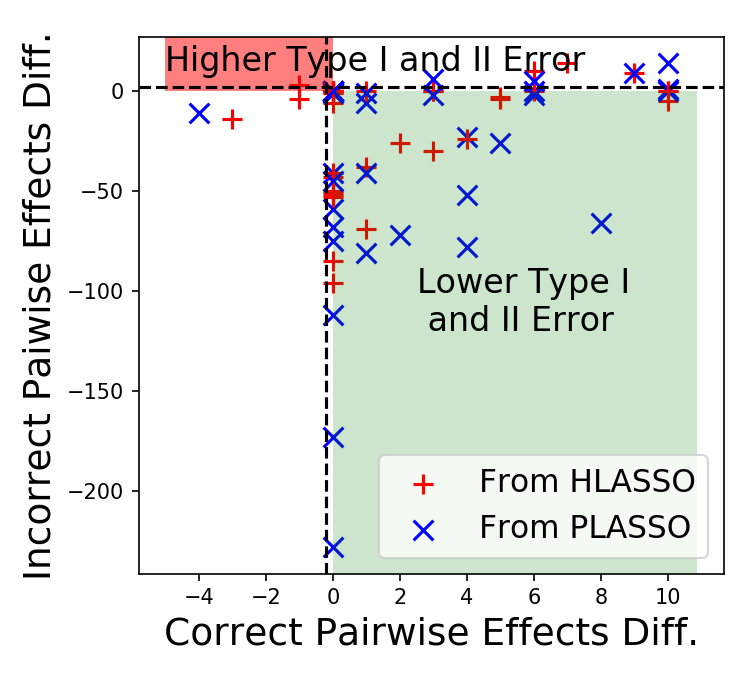}}\par
\end{multicols}
\vspace*{-.2in}
\caption{Variable selection performance of each method for the 36 synthetic datasets. Each point in each plot indicates one of these datasets for a particular method. The green regions in the second and last plot indicate where our method in strictly better than the other two in terms of variable selection, while the red region indicates the datasets for which our method is strictly worse. In the first and last figures, better performance occurs when moving right and/or down.}
\vspace*{-.2in}
\label{fig:var_select}
\end{figure}

\noindent \textbf{Comparison to LASSO: synthetic data.} 
Having demonstrated the considerable computational savings over baseline Bayesian approaches, we next demonstrate the advantage of our method over frequentist approaches such as the LASSO. In particular, we consider the common case when the true high-dimensional parameter $\theta$ is assumed to be sparse and satisfies the requirement of strong hierarchy. 
To the best of our knowledge, there has not been an extensive empirical comparison between sparse Bayesian interaction models and sparse frequentist interaction models. The likely reason is that each MCMC iteration for sampling $\tau$ takes $O(N^2p^2 + N^3)$ time using the Woodbury matrix method. The per-iteration cost of the iterative optimization solver for the LASSO and the hierarchical LASSO, on the other hand, is $O(Np^2)$, which is much faster when $N$ is even moderately large. Fortunately, SKIM admits a cheap-to-compute kernel function such that each MCMC iteration takes $O(N^2p + N^3)$ time, which is \emph{faster} than the LASSO-style approaches in cases when $p$ is large relative to $N$. 

We benchmark SKIM against generating all pairwise interactions and running the LASSO (denoted pairs LASSO) and the hierarchical LASSO \cite{lim_heirch_lasso}, which constrains the fitted parameters to satisfy strong hierarchy. We generate 36 different synthetic datasets, which differ in the number of observations, dimension, and signal-to-noise ratio. The covariates $X$ are drawn from $\mathcal{N}(0, \lambda^{2} I_p)$ for different choices of $\lambda$. Here, $\lambda$ controls the signal-to-noise ratio; when $\lambda$ is larger, the signal is stronger. We consider $N \in \{ 50, 100, 200 \}$ observations, $p \in \{ 50, 100, 200, 500 \}$ dimensions (which translates into between roughly $1.25 \times 10^{3}$ and $1.25 \times 10^5$ total interaction parameters), and $\lambda  \in \{1, 2, 5\}$. 
In each dataset, we select five variables (and their pairwise interactions) to affect $y$, and we allow the rest of the variables to lead to spurious correlations with the response $y$. We set the magnitudes of all non-zero effects to 1. Finally, $y \mid x, \theta^* \sim \mathcal{N}(0, \sigma^2)$, where the noise variance $\sigma^2$ equals the largest $\lambda^{2}$ value, namely $25$, to mimic the realistic case when the noise variance is large relative to the signal.
\begin{figure}[h]
    \centering
    \vspace*{-.2in}
    \hspace*{-.1in}
	\begin{multicols}{2}
    \subfigure[MSE difference (main)]{\includegraphics[width=1.1\linewidth]{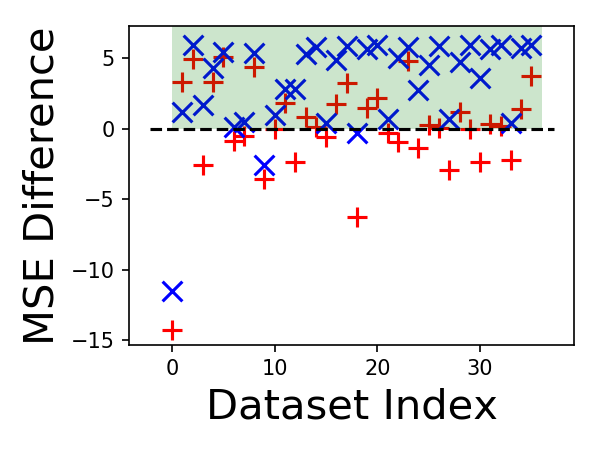}}
    \par
    \hspace*{-.1in}
    \subfigure[MSE difference (pairwise)]{\includegraphics[width=1.1\linewidth]{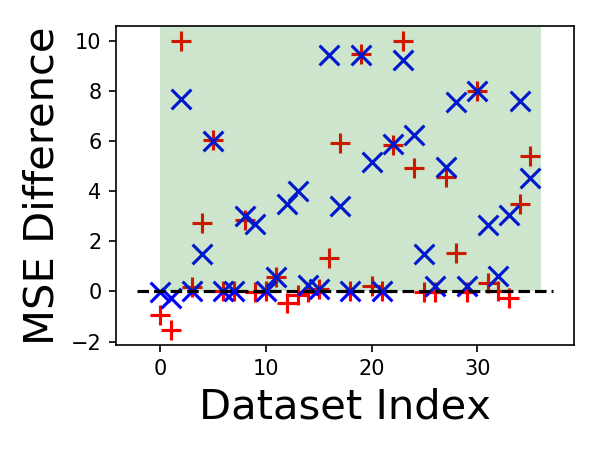}}
    \end{multicols}
    \vspace*{-.2in}
    \caption{Each red cross denotes the difference in MSE of the hierarchical LASSO and KIS from the true main effects (left) and pairwise effects (right) for a given synthetic dataset. When the MSE difference is larger than 0 (i.e., the green shaded region), then our method is closer to the true effect sizes in terms of Euclidean distance. Similarly, each blue x equals the difference in MSE of all-pairs LASSO and our method.}  
    \label{fig:mse_diff}
\end{figure}
We compare each method in terms of variable selection quality and mean-squared error (MSE) between the fitted and true parameter. 
For variable selection, we select a parameter only if the posterior mean of that parameter is farther than 2.59 times its (average) posterior standard deviation from zero; see \cref{A:intervals} for details. For the hierarchical LASSO and pairs LASSO, the variables selected are those with non-zero coefficients, and we use $5$-fold cross-validation to find the strength of the $L_1$ penalties. We fit the hierarchical LASSO using the \texttt{glinternet} package in \texttt{R} and pairs LASSO using \texttt{sklearn} in \texttt{python}. We implemented KIS is \texttt{Stan} (4 chains with 1,000 iterations each). The $\hat{R}$ values for each dataset were less than 1.05. 

First, we examine how well each method selects main effects and pairwise effects. Each point in \cref{fig:var_select}(a) shows the number of main effects selected and number of incorrect main effects selected for a given synthetic dataset. In this plot, it is clear that our method has better false discovery rate (FDR) control over the other two methods on average. \cref{fig:var_select}(c) shows the FDR performance for pairwise effects. To compare the methods at the dataset level, in \cref{fig:var_select}(b,d) we consider the difference in the number of correct and incorrect main effects selected by our method and the LASSO methods for each dataset. The green shaded regions indicate the datasets for which our method simultaneously selects more correct main effects and has fewer incorrect main effects, i.e., is strictly better than the other two methods for any variable selection metric. 
Finally, in \cref{fig:mse_diff} we look at the difference in MSE to $\theta^*$, broken down in terms of the error for estimating main and pairwise effects. Again, we see for the great majority of the datasets, KIS outperforms the LASSO based approaches.  In \cref{fig:runtime_full}(b) we see that SKIM-KIS has competitive runtimes relative to pairs LASSO and hierarchical LASSO.

\noindent \textbf{Comparison to LASSO: synthetic data, real covariates.} To understand the impact of the geometry of the covariates on performance, we took the Residential Building Data Set from the UCI Machine Learning Repository and simulated responses as in our previous synthetic experimental setup. In particular, we randomly chose 5 variables and their 10 pairwise interactions to have non-zero effects. In this case, the covariates are highly correlated (the first 20 out of 105 principal components capture over $99\%$ of the variance in the data). In \cref{table:build_data}, we see that SKIM significantly outperforms the LASSO-based methods for recovering main and pairwise effects. 

\begin{table}[t]
\caption{Building dataset results. MAIN (PAIR) MSE refers to total error in estimating main (pairwise) effects. The main and pairwise MSE added together yield the total MSE. The second and fourth columns show (\# of effects correctly selected) : (\# of incorrect effects selected) for main and pairwise effects, respectively. Larger green values are better while larger purple values are worse.}
\label{table:build_data}
\begin{center}
\begin{small}
\begin{sc}
\begin{tabular}{lcccr}
\toprule
Method & Main MSE & \# Main & Pair MSE &  \# Pair \\ 
\midrule
SKIM    & 0.1 & \textbf{\textcolor{ForestGreen}{3} : \textcolor{Purple}{0}} & 7.0 & \textbf{\textcolor{ForestGreen}{3} : \textcolor{Purple}{0}}  \\
PLASSO & 5.0 & \textbf{\textcolor{ForestGreen}{2} : \textcolor{Purple}{5}}  & 9.3 & \textbf{\textcolor{ForestGreen}{3} : \textcolor{Purple}{21}}  \\
HLASSO    & 1.5 & \textbf{\textcolor{ForestGreen}{3} : \textcolor{Purple}{19}}  & 7.8 & \textbf{\textcolor{ForestGreen}{3} : \textcolor{Purple}{18}}  \\
\bottomrule
\end{tabular}
\end{sc}
\end{small}
\end{center}
\end{table}

\begin{table}[t]
\vspace*{-.1in}
\caption{Auto MPG dataset results. Each column represents the (\# of original effects selected) : (\# of fake effects selected). A selected main (pairwise) effect is an ``original'' effect if it corresponds to one of the original 6 features (15 interactions). Main100 (Pairwise100) and Main200 (Pairwise200) denote when 100 and 200 random noise covariates are added to the original 6 features, respectively. Larger  purple values are worse. Higher green values are not necessarily better since there are no ground truth interactions.}
\label{table:cars}
\begin{center}
\begin{small}
\begin{sc}
\begin{tabular}{lcccr}
\toprule
Method & Main100 & Main200 & Pair100 &  Pair200 \\ 
\midrule
SKIM    & \textbf{\textcolor{ForestGreen}{3} : \textcolor{Purple}{0}} &\textbf{\textcolor{ForestGreen}{3} : \textcolor{Purple}{0}} & \textbf{\textcolor{ForestGreen}{1} : \textcolor{Purple}{0}} & \textbf{\textcolor{ForestGreen}{1} : \textcolor{Purple}{0}} \\
PLASSO & \textbf{\textcolor{ForestGreen}{4} : \textcolor{Purple}{1}} & \textbf{\textcolor{ForestGreen}{4} : \textcolor{Purple}{0}} & \textbf{\textcolor{ForestGreen}{4} : \textcolor{Purple}{99}} & \textbf{\textcolor{ForestGreen}{2} : \textcolor{Purple}{78}} \\
HLASSO    & \textbf{\textcolor{ForestGreen}{5} : \textcolor{Purple}{4}} & \textbf{\textcolor{ForestGreen}{6} : \textcolor{Purple}{46}} & \textbf{\textcolor{ForestGreen}{5} : \textcolor{Purple}{2}} & \textbf{\textcolor{ForestGreen}{4} : \textcolor{Purple}{38}} \\
\bottomrule
\vspace*{-.2in}
\end{tabular}
\end{sc}
\end{small}
\end{center}
\end{table}

\noindent \textbf{Comparison to LASSO: cars miles per gallon dataset.}
We conclude by comparing the methods on the Auto MPG dataset, from the UCI Machine Learning Repository, which contains $N=398$ samples and $p=8$ variables. We consider only the 6 numerical variables (cylinders, displacement, horsepower, weight, acceleration, model year) and standardize the data by subtracting the mean and dividing by the standard deviation. To compare the methods, we first fit SKIM and the LASSO-based methods (via 5-fold cross-validation) on these 6 features. Our method selects three main effects (weight, horsepower, acceleration) and one interaction (weight $\times$ horsepower). The hierarchical LASSO selects all six main effects and 8 out of the 15 possible pairwise interactions. Pairs LASSO selects 5 main effects and 8 interactions. 

Since there is no ground truth, and all of the main and pairwise interactions could a priori affect miles per gallon, it is difficult to compare the methods. To better assess the methods, we instead append random noise covariates and refit each model. In particular, we draw additional covariates from a $\mathcal{N}(0, I_m)$, for $m=100, 200$ and add these noise variables to the original 6 features. The total number of main and pairwise regression coefficients grows to 5,671 and 21,321 for $m=100, 200$ respectively, making the regression task very high-dimensional.  The results are summarized in \cref{table:cars}. All methods are able to pick up some main effects and pairwise effects from the original dataset. Beyond that observation, we cannot compare which main and interaction effects from the original data are real. However, we do know that all noise effects are fake. We see that even with more noise directions, our method selects the same main effects and pairwise effects as the noiseless covariate case; that is, it does not pick up any fake effects. The two LASSO-based methods, on the other hand, incorrectly select many noise variables as interactions.  

\noindent \textbf{Conclusion.} Through our kernel interaction sampler we have demonstrated that Bayesian interaction models can offer both competitive computational scaling relative to LASSO-based methods and improved Type I and II error rates. While our method runs in time linear in $p$ per iteration, the cubic dependence on $N$ still makes inference computationally challenging. Fortunately, there is a wide GP literature that deals precisely with reducing this cubic timing dependence through inducing points \citep{Titsias09variationallearning, Quinonero} or novel conjugate-gradient techniques \citep{black_box}. An interesting future direction will be to empirically and theoretically understand the statistical penalty of using these inducing point methods to scale SKIM to the setting of both large $N$ \emph{and} large $p$.




\clearpage

\section*{Acknowledgements}
This research is supported in part by an NSF CAREER Award, an ARO YIP Award, DARPA, a Sloan Research Fellowship, and ONR.  BLT is supported by NSF GRFP.

\bibliography{references}
\bibliographystyle{icml2019}

\clearpage

\onecolumn
\appendix 

\section{Modeling Multi-Way Interactions} \label{A:higer_degree} In certain applications, we might expect that there are interactions of order greater than two. For example, suppose we are trying to predict college admissions. Then, we might expect a three-way interaction between a candidate's SAT score, GPA, and extracurricular involvement. Individually, these variables might only exhibit moderate association but together they could have a multiplicative effect. For example, we might expect that candidates who have high SAT scores, high GPAs, and  excellent extracurricular activities will be accepted with near certainty, while candidates who only possess one/two of these qualities are borderline applicants.


We now show how to extend our results to handle such three-way, or more generally, \emph{$r$-way interactions}.
\begin{ndefn} ({$r$-way} interactions)
The \emph{$r$-way interactions} of a covariate vector $x \in \R^p$ are generated from the feature map 
\begin{equation*}
\Phi_r(x) \coloneqq \bigoplus_{d=1}^r \ \bigoplus_{k: k_1 + \cdots + k_p = d} \ \prod_{j=1}^p x_j^{k_j}, \quad  k \in \mathbb{N}^p,
\end{equation*} 
where $\bigoplus_{j=1}^m a_j \coloneqq (a_{11}, \cdots, a_{1k_1}, \cdots, a_{m1}, \cdots, a_{mk_m})$ denotes the concatenation of vectors $a_j \in \R^{k_j}$.
\end{ndefn}
To model $r$-way interactions, we must use degree $r$ polynomial kernels to generate all the necessary interactions. Hence, we recommend using the following generalized two-way interaction kernel, which we call the $r$-way interaction kernel.
\begin{ndefn} \label{def:r_way} ($r$-way interaction kernel) A kernel $k$ is called an \emph{$r$-way interaction kernel} if for some choice of $M_1, M_2, M_3 \in \mathbb{N}$, $\alpha, \psi, \lambda^{(m)} \in \R^p_+~(m=1,\dots,M_{1})$, $\nu^{(m)} \in \R_+~(m=1,\dots,M_{2})$,  and $\nabla^{(m)} \in \R^p_+~(k=1,\dots,M_{3})$ it can be re-expressed as
\begin{equation*}
\begin{split}
& \sum_{m=1}^{M_1} k^1_{poly, r}(\lambda^{(m)} \odot x, \lambda^{(m)} \odot y) + \sum_{m=1}^{M_2} \nu^{(m)} \left[ \prod_{s=1}^r x_{i_{s_m}} \prod_{s=1}^r y_{i_{s_m}} \right] +  \sum_{m=1}^{M_3} k_{r-1}(\nabla^{(m)}  \odot x, \nabla^{(m)}  \odot y),
\end{split}
\end{equation*}
where $\odot$ is the Hadamard product and $k_{r-1}$ is an $r-1$ degree interaction kernel. The base case kernel (i.e., when $r=2$) is provided in \cref{def:two_way}.
\end{ndefn}
To select the weights for an $r$-way interaction kernel, we must solve a system of equations similar to \cref{eq:system_eq}, except for a target prior covariance matrix $\Sigma_{\tau} \in \R^{\dim(\Phi_r) \times \dim(\Phi_r)}$.

\section{Proofs} \label{A:proofs}
\subsection{Proof of \cref{prop:explicit_kernel}}
Let $g(\cdot) = \theta^T \Phi_2(\cdot)$ and $\theta \mid \tau \sim \mathcal{N}(0, \Sigma_{\tau})$. Then, $y^{(n)} = g(x^{(n)}) + \epsilon^{(n)}$.  The first claim follows by taking $\phi = \Phi_2$ and $f = g$ in \citet[Equation 2.12]{gp_book}.

The second claim follows directly from the duality between the weight-space and function-space view of a GP \citep[Chapter 2]{gp_book}.

%
%
%
%
%
%
%

\subsection{Proof of \cref{thm:induced_prior}} \label{A:mult_kern_trick}

The proof of \cref{thm:induced_prior} depends critically on \cref{lem:sum_kern} below, which characterizes the relation between adding two kernels and the resulting induced prior covariance matrix.

\begin{nlem} \label{lem:sum_kern}
Let $k_1$ and $k_2$ be two kernels such that there exists vectors $a^{(1)}, a^{(2)} \in \R^{\text{dim}(\Phi_2)}$ for which $k_i(x, y) = \inner{a^{(i)} \odot \Phi_2(x)}{a^{(i)} \odot \Phi_2(y)} $. Let $k_3(x, y) = k_1(x, y) + k_2(x, y)$. Then,
\begin{equation}
k_3(x, y) = \inner{\Sigma_3^{\frac{1}{2}}  \Phi_2(x)}{\Sigma_3^{\frac{1}{2}}  \Phi_2(y)} \quad \text{s.t.} \quad \Sigma_3 = \text{diag}(a^{(1)} \odot  a^{(1)} + a^{(2)} \odot  a^{(2)} ).
\end{equation} 
\end{nlem}
\begin{proof}
By the sum property of kernels,
\begin{equation}
\begin{split}
k_1(x, y) + k_2(x, y) &= \inner{[a_1 \ a_2] \odot [\Phi_2(x) \ \Phi_2(x) ]}{[a_1 \ a_2] \odot [\Phi_2(y) \ \Phi_2(y) ]} \\
							& =  \inner{a^{(1)} \odot \Phi_2(x)}{a^{(1)} \odot \Phi_2(y)} + \inner{a^{(2)} \odot \Phi_2(x)}{a^{(2)} \odot \Phi_2(y)} \\
							& = \inner{a^{(1)} \odot  a^{(1)} \odot \Phi_2(x)}{\Phi_2(y)} + \inner{a^{(2)} \odot  a^{(2)} \odot \Phi_2(x)}{\Phi_2(y)} \\
							& = \inner{a^{(1)} \odot  a^{(1)} \odot \Phi_2(x) + a^{(2)} \odot  a^{(2)} \odot \Phi_2(x)}{\Phi_2(y)} \\
							& = \inner{(a^{(1)} \odot  a^{(1)} + a^{(2)} \odot  a^{(2)} )  \odot \Phi_2(x)}{\Phi_2(y)} \\ 
							& = \Phi_2^T(x) \ \text{diag}((a^{(1)} \odot  a^{(1)} + a^{(2)} \odot  a^{(2)} ) \ \Phi_2(y) \\
							& = k_3(x, y).
\end{split}
\end{equation}
\end{proof}
By \cref{lem:sum_kern}, it suffices to write out the feature map of each kernel in \cref{def:two_way}. The induced feature maps of each respective kernel term in \cref{def:two_way} are given by $a_i \odot \Phi_2(x), 1 \leq i \leq 4$ for
\begin{equation} \label{eq:feat_maps}
\begin{split}
a_1 & \coloneqq ((\lambda^{(m)}_1)^2 , \cdots, (\lambda^{(m)}_p)^2, \sqrt{2} \lambda^{(m)}_1 \lambda^{(m)}_2, \cdots, \sqrt{2} \lambda^{(m)}_{p-1} \lambda^{(m)}_p, \sqrt{2} \lambda^{(m)}_1, \cdots, \sqrt{2} \lambda^{(m)}_p, 1) \\
a_2 & \coloneqq(0, \cdots, 0, 0, \cdots, 0, \alpha_1, \cdots, \alpha_p, \sqrt{A}) \\
a_3 & \coloneqq (\psi_1, \cdots, \psi_p, 0, \cdots, 0, 0, \cdots, 0, 0) \\
a_4 & \coloneqq (0, \cdots, 0, 0, \cdots, 0, \sqrt{\nu^{(m)}}, 0, \cdots, 0, 0, \cdots, 0, 0)
\end{split}
\end{equation} 
The first claim follows from \cref{eq:feat_maps} and \cref{lem:sum_kern}. 

To prove the second claim, take an arbitrary diagonal prior covariance matrix $S \in \R^{\dim(\Phi_2) \times \dim(\Phi_2)} $. It suffices to show that there exists a solution of,
\begin{equation*}
\begin{split}
 \diag(S)_{(i)}  &= \alpha_i^2 + 2 \sum_{m=1}^{M_1} \left[ \lambda^{(m)}_i \right]^2 \\
  \diag(S)_{(ij)} &= 2 \sum_{m=1}^{M_1} \left[ \lambda^{(m)}_i  \lambda^{(m)}_j \right]^2 + \sum_{m: i_m = i, j_m=j}^{K_2} \nu^{(m)}  \\
 \diag(S)_{(ii)} &= \psi_i^2 + \sum_{m=1}^{M_1} \left[ \lambda^{(m)}_i  \right]^4 \\
 \diag(S)_{(0)} &= M_2 + A.
\end{split}
\end{equation*}
for some choice of $M_1, M_2 \in \mathbb{N}$, $\alpha, \psi, \lambda^{(m)} \in \R^p_+~(m=1,\dots,M_{1})$, $\nu^{(m)} \in \R_+~(m=1,\dots,M_{2})$, and $A \in \R$. Take $\alpha_i^2 = \diag(S)_{(i)} $ and $\psi_i^2 = \diag(S)_{(ii)} $, for $i=1, \cdots, p$. Take $\lambda^{(m)} = 0$. Let $M_2 = \frac{p(p-1)}{2}$ and $\nu^{(1)} = \diag(S)_{(12)}, \cdots, \nu^{(M_2)} = \diag(S)_{((p-1)p)}$. Finally, letting $A = \diag(S)_{(0)} - M_2$ solves the system.

\begin{rmk}
While we have shown one of the \emph{many} ways to solve the above system for an arbitrary $S$, the strategy taken above is not practically useful; computing the kernel in this fashion will take $\Theta(p^2)$ time because $M_2 = \Theta(p^2)$. In practice, we must leverage the polynomial kernels (i.e., those in the $M_1$ sum) to avoid making $M_2$ large. We show how such a strategy works in \cref{A:example_models}.
\end{rmk}

\subsection{Proof of \cref{prop:formula_marginal}}
Define $g(A^{ij}) \coloneqq (g(e_i), g(-e_i), g(e_j), g(e_{ij}))$. Then,
\begin{equation} \label{eq:joint_dist}
\begin{split}
& g(A^{ij}) \mid D, \tau \sim \mathcal{N}(\mu_{g_{ij}}, \Sigma_{ij}) \quad \text{s.t.} \quad \mu_{g_{ij}} \coloneqq  K_{\tau}(A^{ij}, X) H_{\tau} Y, \\
	& \qquad \Sigma_{ij} \coloneqq  \left[K_{\tau}(A^{ij}, A^{ij}) - K_{\tau}(A^{ij}, X) H_{\tau} K_{\tau}(X, A^{ij})  \right], 
\end{split}
\end{equation}
which follows directly from  \citet[Equation 2.21]{gp_book}. Notice that,
\begin{equation} \label{eq:test_point_select}
\theta_{x_i} = \frac{g(e_1)}{2} - \frac{g(-e_1)}{2} = a_i^T g(A^{ij}) \quad \text{and} \quad \theta_{x_ix_j} =  \frac{g(e_1)}{2} - \frac{g(-e_1)}{2} - g(e_j) + g(e_{ij}) = a_{ij}^T g(A^{ij}),
\end{equation}
where $a_i = (\nicefrac{1}{2}, -\nicefrac{1}{2}, 0, 0)$ and $a_{ij} = (-\nicefrac{1}{2}, \nicefrac{1}{2}, -1, 1)$. Furthermore, $ \theta_{x_i^2} = a_{ii}^T g(A^{ij})$ for $a_{ii} = (\nicefrac{1}{2}, \nicefrac{1}{2}, 0, 0)$.  The proof follows from \cref{eq:joint_dist}, \cref{eq:test_point_select}, and recalling that an affine transformation $h: x \mapsto Ax$ of a multivariate Gaussian distribution $Z \sim \mathcal{N}(\mu, \Sigma)$ is given by $h(Z) \sim \mathcal{N}( A \mu,   A \Sigma A^T)$.

\subsection{Proof of \cref{cor:get_effects_time}} \label{A:pf_cor_effects}
\cref{cor:get_effects_time} follows immediately once we can show that $K_{\tau}( A_{ij}, X)$ takes $O(1)$ time. It suffices to show $k_{\tau}(x^{(n)}, e_i)$ and $k_{\tau}(x^{(n)}, e_i + e_j)$ take $O(1)$ time. Since $k_{\tau}$ is a sum of polynomial kernels, $k_{\tau}(x, y)$ only depends on $x, y \in \R^p$ through the inner product $x^Ty$. Hence, for vectors $\tilde{x}, \tilde{y} \in \R^M$, $k_{\tau}(\tilde{x}, \tilde{y})$ is well-defined and just depends on $\tilde{x}^T\tilde{y}$. Now, $k_{\tau}(x^{(n)}, e_i) = k_{\tau}(x^{(n)}_i, 1)$ and $k_{\tau}(x^{(n)}, e_i + e_j) = k_{\tau}((x^{(n)}_i, x^{(n)}_j), (1, 1))$. Since $k_{\tau}(x^{(n)}_i, 1)$ and $k_{\tau}((x^{(n)}_i, x^{(n)}_j), (1, 1))$ do not depend on $p$, these terms each take $O(1)$ time to compute. 

\subsection{The General Kernel Interaction Trick} \label{A:the_mult_kern_trick}
In this section, we generalize the kernel interaction trick, namely show how to access the distribution of arbitrary components of $\theta$. First, we require some new notation. For $E \subseteq \{1, \cdots, p\}$, $|E| = M$, define 
\begin{equation}
\theta_E \coloneqq (\theta_{x_{i_1}}, \cdots, \theta_{x_{i_M}}, \theta_{x_{i_{1}}  x_{i_{2}}}, \cdots, \theta_{x_{i_{M-1}}  x_{i_{M}}}), \quad i_j \in E.
\end{equation}
We show how to compute $\theta_E \mid \tau, D$ from the GP posterior predictive distribution. Without any lost of generality, we may assume $E = \{1, \cdots, M \}$ by relabeling the covariates.

\begin{nthm} (General kernel interaction trick) \label{prop:gen_kern_trick}
Let $H_{\tau} \coloneqq (K_{\tau} + \sigma^2 I_N)^{-1}$ and
\begin{equation*}
    A_M \coloneqq [e_1, -e_1, \cdots e_{M}, -e_{M}, e_{1} + e_{2}, \cdots, e_{M-1} + e_{M}]^T.
\end{equation*}
Let $K_{\tau}(A_M, X) = K_{\tau}(X, A_M)^T$ be the matrix formed by taking the kernel between each row of $A_M$ with each row of $X$. Let 
\begin{equation}
\begin{split}
a_i & \coloneqq (0, 0, \cdots, \nicefrac{1}{2}, \nicefrac{-1}{2}, \cdots, 0, 0, \cdots, 0) \in \R^{2M + \frac{M(M-1)}{2}} \\
a_{ij}  & \coloneqq (0, 0, \cdots, \nicefrac{1}{2}, \nicefrac{-1}{2}, \cdots, -1, \cdots, 0, 0, \cdots, 1, \cdots, 0) \in \R^{2M + \frac{M(M-1)}{2}}
\end{split}
\end{equation}
for $i < j$. That is, $a_i$ has non-zero entries at $e_i$ and $-e_i$ and $a_{ij}$ has non-zero entries at $e_i$, $-e_i$, $-e_j$, and $e_i + e_j$. Let
\begin{equation}
R_M \coloneqq [a_1 \cdots a_M \ \ a_{12} \cdots a_{(M-1)M}]^T.
\end{equation}
Then, $\theta_E \mid \tau, D$ is a multivariate Gaussian distribution with mean $R_M K_{\tau}(A_M, X) H_{\tau} Y$  and covariance matrix 
\begin{equation*}
    R_M \left[K_{\tau}(A_{ij}, A_{ij}) - K_{\tau}(A_{ij}, X) H_{\tau} K_{\tau}(X, A_{ij})  \right]R_M^T.
\end{equation*}
\end{nthm}

\begin{proof}

Following the proof of \cref{prop:formula_marginal},
\begin{equation} \label{eq:mult_joint_dist}
\begin{split}
& g(A^{M}) \mid D, \tau \sim \mathcal{N}(\mu_{g_{M}}, \Sigma_{M}) \quad \text{s.t.} \quad \mu_{g_{M}} \coloneqq  K_{\tau}(A^{M}, X) H_{\tau} Y, \\
	& \qquad \Sigma_{M} \coloneqq  \left[K_{\tau}(A^{M}, A^{M}) - K_{\tau}(A^{M}, X) H_{\tau} K_{\tau}(X, A^{M})  \right]. 
\end{split}
\end{equation}
Similar to \cref{eq:test_point_select},
\begin{equation} \label{eq:mult_test_point_select}
\theta_{x_i} = \frac{g(e_1)}{2} - \frac{g(-e_1)}{2} = a_i^T g(A^{M}) \quad \text{and} \quad \theta_{x_ix_j}  =  \frac{g(e_1)}{2} - \frac{g(-e_1)}{2} - g(e_j) + g(e_{ij}) = a_{ij}^T g(A^{M}).
\end{equation}
The proof follows from \cref{eq:mult_joint_dist}, \cref{eq:mult_test_point_select}, and recalling that an affine transformation $h: x \mapsto R^T_Mx$ of a multivariate Gaussian distribution $Z \sim \mathcal{N}(\mu, \Sigma)$ is given by $h(Z) \sim \mathcal{N}( R_M \mu,   R_M \Sigma R^T_M)$.

\end{proof}

\begin{ncor} \label{cor:general_kern}
Given $K_{\tau}$, the distribution $\theta_E \mid \tau, D$ takes $O(M^2)$ time and memory to compute.
\end{ncor}
\begin{proof}
The proof is identical to the one provided in \cref{A:pf_cor_effects}.
\end{proof}

\subsection{Proof of \cref{prop:pairwise_sparse_kern}} \label{A:proof_of_skim}
See \cref{A:sparse_prior}.

\section{Example Bayesian Interaction Models} \label{A:example_models} In the following subsections, we show how to solve \cref{eq:system_eq} for several classes of models.

\subsection{Block-Degree Priors}
Suppose we would like to set the prior variance of all terms with the same degree equal. That is, we would like to use a prior of the form
\begin{equation} \label{eq:block_prior}
\begin{split}
\eta \in \R^3  &\sim p(\eta) \\
\theta_{x_i} \mid \eta & \sim \mathcal{N}(0, \eta_1^2) \\
\theta_{x_ix_j} \mid \eta & \sim \mathcal{N}(0, \eta_2^2) \\
\theta_{x_i^2} \mid \eta & \sim \mathcal{N}(0, \eta_3^2) \\
\theta_{0} \mid c^2 & \sim \mathcal{N}(0, c^2).
\end{split}
\end{equation}
To find the corresponding kernel, let $\lambda = (\frac{1}{\sqrt[4]{2}}\sqrt{\eta_2}, \cdots, \frac{1}{\sqrt[4]{2}}\sqrt{\eta_2})$, $M_1 = 1$ and $M_2 = 0$. Then, $\diag(S)_{(ij)} = \eta_2^2$. Setting $\psi_i^2 =  \eta_3^2 - \frac{1}{2}\eta_2^2$, implies that $\diag(S)_{(ii)} = \eta_3^2$. Finally, letting $\alpha_i^2 = \tau_1^2 - \frac{2 \eta_2}{\sqrt{2}}$ and $A = c^2 - 1$ implies that $\diag(S)_{(i)} = \eta_1^2$ and $\diag(S)_{(0)} = c^2$ as desired. We may equivalently re-write the induced kernel as
\begin{equation} 
   k_{block, \eta}(x, y) = \frac{\eta_2^2}{2} k_{\text{poly}, 2}^{1}(x, y) + (\eta_3^2 - \frac{\eta_2^2}{2})k_{\text{poly}, 1}^{0}( x \odot x, y \odot y) + \left (\eta_1^2 - \eta_2^2 \right) k_{\text{poly}, 1}^{0}(x , y) + c^2 - \frac{\eta_2^2}{2}.
\end{equation}
Hence, \cref{eq:block_prior} admits a kernel that only takes $O(p)$ time to compute.
\subsection{Sparsity Priors} \label{A:sparse_prior}
By \cref{lem:sum_kern}, the sparsity prior model provided in \cref{eq:hierch_prior} equals $k_{block, \eta}(\kappa \odot x, \kappa \odot y)$.


 
\section{SKIM Model Details} \label{A:skim} We provide the full hierarchical form of SKIM next. SKIM is based closely on the \emph{regularized horseshoe prior} \cite{finnish_prior} and the model proposed in \citet{heirc_sparisty}:
\begin{align*}
\begin{split}
    m^2 \sim \text{InvGamma}(\alpha_1, \beta_1) & \qquad \xi^2 \sim \text{InvGamma}(\alpha_2, \beta_2) \\
   \psi^2 \sim \text{InvGamma}(\alpha_2, \beta_2) &\qquad \phi \coloneqq \frac{s}{p - s} \frac{\sigma}{\sqrt{N}} \qquad \sigma \sim \mathcal{N}^+(0, \alpha_3) \\
   \kappa_i = \frac{m \lambda_i}{\sqrt{m^2  + \eta_1^2 \lambda_i^2}}  & \qquad \lambda_i \sim C^{+}(0, 1) \\
   \eta_1 \sim C^+(0, \phi)  & \qquad \eta_2 = \frac{\eta_1^2}{m^2} \xi \qquad \eta_3 = \frac{\eta_1^2}{m^2} \psi  \nonumber
\end{split} \\
\begin{split}
\theta_{x_i} \mid \eta, \kappa &\sim \mathcal{N}(0, \eta_1^2 \kappa_i^2)  \\
\theta_{x_j} \mid \eta, \kappa &\sim \mathcal{N}(0, \eta_1^2 \kappa_j^2)  \\
\theta_{x_i x_j} \mid \eta, \kappa &\sim \mathcal{N}(0, \eta_2^2 \kappa_i^2 \kappa_j^2) \\
\theta_{x_i^2} \mid \eta, \kappa &\sim \mathcal{N}(0, \eta_3^2 \kappa_i^4) \\
\theta_0 \mid c^2 &\sim \mathcal{N}(0, c^2),
\end{split}
\end{align*}
where $s, \alpha_i$, and $\beta_i$ are user-specified hyperparameters, $C^+(0, 1)$ is a {half-Cauchy} distribution, and $\mathcal{N}^+$ is a {half-normal} distribution. In \cref{sec:experiments}, we set $s=5$, $\alpha_1 = 12.5$,  $\alpha_2 = 12.5$,  $\alpha_3 = 2$, $\beta_1 = 112.5$, and $\beta_2 = 12.5$. 
More details, such as selecting the hyperparameters, desirable properties, and interpretations of SKIM, are provided below.

\subsection{SKIM Details} \label{A:skim_details}
Recall that we are primarily interested in the case when $\theta$ is sparse and satisfies strong-hierarchy. In order to promote sparsity in the main effects, we require two ingredients: (1) a prior on the \emph{global shrinkage} parameter $\eta_1$ and (2) a prior on the \emph{local shrinkage} parameters contained in $\kappa \in \R^p$ \cite{horse_prior, finnish_prior}.  Conditional on $\eta_1$ and $\kappa$, 
\begin{equation} \label{eq:main_prior_var}
\theta_{x_i} \mid \kappa, \eta_1 \sim \mathcal{N}(0, \eta_1^2 \kappa_i^2), \quad i = 1, \cdots, p.
\end{equation}
$\eta_1$ controls the overall sparsity level of the model; in particular, the model becomes sparser as $\eta_1$ decreases. If we expect $s$ non-zero main effects, then setting $\eta_1 = \frac{s}{p - s} \frac{\sigma}{\sqrt{N}}$ will yield an expected prior sparsity level of $s$ by \citet[Equation 3.12]{finnish_prior}. However, we often do not know exactly how to select $s$. Hence, \citet{finnish_prior} instead  draw,
\begin{equation}
 \phi \coloneqq \frac{s}{p - s} \frac{\sigma}{\sqrt{N}} \qquad  \eta_1 \sim C^+(0, \phi),
\end{equation}
to express the uncertainty of not knowing the true main effect sparsity level. 

The prior variance of $\theta_{x_i}$ is non-negligible only when $\kappa_i$ is large enough to escape the global shrinkage of $\eta_1$. Hence, we draw $\kappa_i$ from a heavy-tailed distribution so that certain main effects can escape global shrinkage. \citet{horse_prior} suggest drawing $\kappa_i$ from a half-Cauchy distribution since this distribution has fat tails and desirable shrinkage properties. However, such a prior often leads to undesirable numerical stability issues when using gradient-based MCMC methods such as NUTS \citep{finnish_prior}. As a result, \citet{finnish_prior} instead propose the \emph{regularized horseshoe} prior, which truncates the half-Cauchy distribution to have support only on $[0, m)$ instead of $[0, \infty)$. This truncation (empirically) leads to better mixing properties, and is achieved by setting
\begin{equation}
\kappa_i = \frac{m \lambda_i}{\sqrt{m^2  + \eta_1^2 \lambda_i^2}},  \qquad \lambda_i \sim C^{+}(0, 1).
\end{equation}
As $\lambda_i \rightarrow \infty$, $\kappa_i \rightarrow \frac{m}{\eta_1}$. Hence, as $\lambda_i \rightarrow \infty$, the prior variance of $\theta_{x_i}$ equals $m$. Since we might not know the scale $m$ of the non-zero main effects, we place a prior on $m$, namely,
\begin{equation}
 m^2 \sim \text{InvGamma}(\alpha_1, \beta_1)
\end{equation}
for hyperparameters $\alpha_1$ and $\alpha_2$.
 
Next, we model the interactions. If strong-hierarchy holds, sparsity comes for free; if there are only $s \ll p$ non-zero main effects, then there are at most $\frac{s(s-1)}{2} \ll p^2$ possible pairwise interactions. We must be careful, however, because strong-hierarchy trivially holds; our main effect estimates will, with probability one, never equal zero because the prior variances of the main effects are greater than $0$ with probability one by \cref{eq:main_prior_var} and our choice of priors. Instead, we aim for a relaxed version of strong-hierarchy. Namely, that the prior variance of an interaction $\theta_{x_i x_j}$ is large only if $\theta_{x_i}$ and $\theta_{x_j}$ are both large. Notice that the prior variances on $\theta_{x_i}$ and $\theta_{x_j}$ are large only when $\kappa_i$ and $\kappa_j$ are sufficiently far from zero. Hence, it suffices to make the prior variance of $\theta_{x_ix_j}$ large only when $\kappa_i$ and $\kappa_j$ are both large. 

Let $\tilde{\kappa}_i^2 = \frac{\eta_1^2}{m^2} \kappa_i^2$. Then, $0 \leq \tilde{\kappa}_i^2 \leq 1$ and $\tilde{\kappa}_i$ approaches $1$ as $\lambda_i \rightarrow \infty$. Since, $\tilde{\kappa}_i^2$ and $\tilde{\kappa}_j^2$ are bounded by 1, $\tilde{\kappa}_i^2 \tilde{\kappa}_j^2$ will only be close to 1 when each term is close to one. That is, when both  $\lambda_i$ \emph{and} $\lambda_j$ are large, or equivalently when $\kappa_i$ and $\kappa_j$ are both large. Hence, it suffices to let 
\begin{equation}
\begin{split}
\theta_{x_i x_j} \mid \eta_1, \kappa &\sim \mathcal{N}(0, \xi^2  \tilde{\kappa}_i^2  \tilde{\kappa}_j^2) \\
& = \mathcal{N}(0, \eta_2^2 \kappa_i^2 \kappa_j^2) \quad \text{for} \quad \eta_2 \coloneqq  \frac{\eta_1^2}{m^2} \xi,
\end{split}
\end{equation}
to promote strong-hierarchy, where $\xi$ has the interpretation of the scale of the  non-zero interaction effects; as $\lambda_i$ and $\lambda_j$ tend to infinity, the prior variance of $\theta_{x_i x_j}$ approaches $\xi^2$. Since we might not know this scale, we draw
\begin{equation}
 \xi^2 \sim \text{InvGamma}(\alpha_2, \beta_2),
\end{equation}
for some choice of hyperparameters $\alpha_2$ and $\beta_2$. Our choice of prior for $\theta_{x_i^2}$ is analogous to the above reasoning for the choice of prior on $\theta_{x_ix_j}$.

\textbf{Main difference between SKIM and the model proposed in \citet{heirc_sparisty}}: Unlike in the model proposed in \citet{heirc_sparisty}, SKIM does not assume sparsity between the interactions once the main effects are known. In particular, suppose, without any loss of generality, that the first $s$ components of $\lambda$ are large, while the remaining $p-s$ components are very close to zero. Then, the only interactions with non-negligible prior variance are the interactions between the first $s$ variables. The number of such interactions is $O(s^2)$. 

Unlike in \citet{heirc_sparisty}, SKIM does not assume sparsity among these $O(s^2)$ interactions. We do not assume such sparsity because the true sparsity level $s$ is often very small (e.g., as in genome-wide association studies), which means that $s^2$ is small. Hence, once we have identified which of the $s$ variables have non-zero main effects, estimating $O(s^2)$ interactions from $N$ datapoints is not statistically difficult relative to actually identifying the $s$ non-zero main effects. In particular, the mean-squared error of estimating $O(s^2)$ parameters from $N$ datapoints is $O\left(\sqrt{\frac{s^2}{N}}\right)$ by standard Bernstein-von Mises results and a union bound. Thus, if $s = o(\sqrt{N})$, we can accurately estimate $O(s^2)$ parameters. 

\section{Variable Selection Procedure} \label{A:intervals} Suppose we sample $\tau^{(t)} \sim p(\tau \mid D)$ via our kernel interaction sampler. Then, we use these $\tau^{(t)}$ samples to perform variable selection in the following way. Without any loss of generality, suppose we are deciding whether or not to include the main effect  $\theta_{x_i}$. Below we will show how to construct an interval $(c_{\text{lower}}, c_{\text{upper}})$ for $\theta_{x_i}$. If this interval does not contain zero, we select $\theta_{x_i}$. This interval is constructed by averaging the posterior means and standard deviations of $\theta_{x_i}$ associated with each sampled $\tau^{(t)}$:  
%
%
\begin{equation*}
\mu_T \coloneqq \frac{1}{T}\sum_{t=1}^T \E_{p(\theta_{x_i} \mid D, \tau^{(t)})} [ \theta_{x_i}] \quad \sigma_T \coloneqq \frac{1}{T}\sum_{t=1}^T \text{SD}_{p(\theta_{x_i} \mid D, \tau^{(t)})} [ \theta_{x_i}]
\end{equation*}
\begin{equation} \label{eq:interval}
c_{\text{lower}} \coloneqq \mu_T - z \sigma_T \quad c_{\text{upper}} \coloneqq \mu_T + z \sigma_T. 
\end{equation}
Here, $\text{SD}_{q(\theta)}[\theta]$ denotes the standard deviation of $\theta$ with respect to $q(\theta)$. In our experiments, we set $z = 2.59$ which corresponds to the $99.5$th percentile of a standard normal distribution.

\textbf{Heuristic justification of our variable selection procedure.} 



We might expect that the posterior $p(\theta_{x_i} \mid D)$ has two modes: one mode near zero when the prior variance of $\theta_{x_i}$ is small and another mode when the prior variance is large. Thus, the posterior mean $\mu_T$ will ``shrink'' the estimate of $\theta_{x_i}$ towards zero, where the amount of shrinkage depends on the posterior mass of each mode. To understand the variability of the posterior mean, we effectively average the variability within each mode in \cref{eq:interval}. This averaging of variability within modes has the advantage of not artificially increasing the variance (due to the modes being separated by regions of low-probability) but has the disadvantage of potentially underestimating our uncertainty. For example, suppose $\theta_{x_i} \mid D \overset{d}{=} .5 \mathcal{N}( 0, .1) + .5 \mathcal{N}(2, .1)$. Then, $\E_{p(\theta_{x_i} \mid D)}[\theta_{x_i}] = 1$ and  $\text{SD}_{p(\theta_{x_i} \mid D)}[\theta_{x_i}] = 1$. In this case, we would not select $\theta_{x_i}$ if we required that the posterior mean is further than twice the posterior standard deviation. If we instead averaged the variance within the modes (which would equal .1), we would select $\theta_{x_i}$ as we do in \cref{eq:interval}. 

While we found good empirical performance of our variable selection procedure in \cref{sec:experiments} (e.g., based on FDR) we nevertheless think that variable selection in multimodal posteriors is challenging, and an area of active research. An interesting future research direction would be to develop even better variable selection strategies for sparse interaction models or to rigorously understand the tradeoffs between different variable selection procedures.

\section{Woodbury Identity and the Matrix Determinant Lemma} \label{sec:woodbury} To compute the determinant in \cref{eq:norm_constant}, one can perform a Cholesky decomposition of $\Sigma_{N, \tau} \in \R^{\dim(\Phi_2) \times \dim(\Phi_2)}$. Computing $\Sigma_{N, \tau}$ takes $O(p^4N)$ time and $O(p^2N + p^4)$ memory. Computing the Cholesky decomposition of $\Sigma_{N, \tau}$ takes $O(p^{6})$ time and  requires $O(p^{4})$ memory. This factorization can be avoided through the {Woodbury matrix lemma} and matrix determinant lemma.

The Woodbury matrix identity implies that
\begin{equation}
(A^{-1} + U U^T)^{-1} = A - A U (I_K + U^T A U)^{-1} U^T A,
\end{equation}
where $A \in \R^{M \times M}$, $U \in \R^{M \times K}$, and $I_K$ is the $K \times K$ identity matrix.
The matrix determinant lemma implies that
\begin{equation}
\det(A^{-1} + U U^T) = \det(I + U^T A U) \det(A^{-1}).
\end{equation}
Then, by the Woodbury identity,
\begin{equation}
\Sigma_{\tau, N} = (\Sigma_{\tau}^{-1} + \frac{1}{\sigma^2} \Phi_2(X)^T \Phi_2(X))^{-1} = \Sigma_{\tau} - \Sigma_{\tau} \Phi_2(X)^T (I_N + \Phi_2(X) \Sigma_{\tau} \Phi_2(X)^T)^{-1}  \Phi_2(X) \Sigma_{\tau}. 
\end{equation}
Computing $p( D \mid \tau)$ requires computing $\det(\Sigma_{\tau, N})$. By the matrix determinant lemma, 
\begin{equation}
\det(\Sigma_{\tau, N}) = (\det(I_N + \Phi_2(X) \Sigma_{\tau} \Phi_2(X)^T) \det(\Sigma_{\tau}^{-1}))^{-1}.
\end{equation}
When $\Sigma_{\tau}$ is diagonal, the determinant  equals the product of the diagonal, and its inverse equals one over the diagonal. Both of these quantities can be computed in $O(p^2)$ time. Hence, the time complexity for computing $\det(\Sigma_{\tau, N})$ is dominated by computing $\det(I_N + \Phi_2(X) \Sigma_{\tau} \Phi_2(X)^T)$, which takes  $O(N^2p^2 + N^3)$ time and $O(Np^2)$ memory to store $\Phi_2(X)$.

\section{Standard Polynomial Kernel} \label{A:poly_kernel}
The feature map induced by the standard degree two polynomial kernel is given by
\begin{equation} \label{eq:poly_feat_map2}
\begin{split}
\Phi_{\text{poly}, 2}^c(x) &\coloneqq (x_1^2, \cdots, x_p^2, \sqrt{2} x_1 x_2, \cdots, \sqrt{2} x_{p-1} x_p, \sqrt{2c} x_1, \cdots, \sqrt{2c} x_p, c) \\
				&= a_{poly, 2} \odot \Phi_2(x), \ a_{poly, 2} \coloneqq (1, \cdots, 1, \sqrt{2}, \cdots, \sqrt{2}, \sqrt{2c}, \cdots, \sqrt{2c}, c). 
\end{split}
\end{equation}
Hence, \cref{eq:poly_feat_map2} implies that
\begin{equation} \label{eq:poly_prior_cov}
\text{diag}(\Sigma_{poly, 2}) =  a_{poly, 2} \odot  a_{poly, 2}.
\end{equation}
\cref{eq:poly_prior_cov} shows that the prior covariance of the interaction terms are given higher prior variance than the main effects when $c \leq 1$, which is often undesirable. Furthermore, this prior does not promote sparsity, which is typically expected in high-dimensional problems.


\end{document}